\documentclass[11pt]{article}
\usepackage[a4paper, margin=1in]{geometry}
\usepackage{amsmath}
\usepackage{amsfonts}
\usepackage{amssymb}
\usepackage{amsthm}
\usepackage{graphicx}

\DeclareMathOperator*{\E}{\mathbb{E}}

\newcommand{\LSH}{\mathcal{H}}

\newcommand{\cube}[1]{\{-1,1\}^{#1}}
\newcommand{\ip}[2]{\langle{#1},{#2}\rangle}

\newtheorem{theorem}{Theorem}
\newtheorem{lemma}[theorem]{Lemma}

\theoremstyle{definition}
\newtheorem{definition}[theorem]{Definition}

\theoremstyle{remark}

\begin{document}

\title{Optimal Boolean Locality-Sensitive Hashing}

\author{
	Tobias Christiani\\
	\small \texttt{tobc@itu.dk}\\
	\small IT University of Copenhagen and BARC
}
\maketitle

\begin{abstract}
For $0 \leq \beta < \alpha < 1$ the distribution $\LSH$ over Boolean functions $h \colon \cube{d} \to \{-1, 1\}$ that minimizes the expression 
\begin{equation*} 
	\rho_{\alpha, \beta} = \frac{\log(1/\Pr_{\substack{h \sim \LSH \\ (x, y) \text{ $\alpha$-corr.}}}[h(x) = h(y)])}{\log(1/\Pr_{\substack{h \sim \LSH \\ (x, y) \text{ $\beta$-corr.}}}[h(x) = h(y)])}
\end{equation*}
assigns nonzero probability only to members of the set of dictator functions $h(x) = \pm x_i$. 
\end{abstract}

\section{Introduction}
We will be studying Boolean functions, i.e., functions that for a positive integer $d$ can be written in the form
\begin{equation*}
	h \colon \{-1, 1\}^d \to \{-1, 1\}.
\end{equation*}
We are concerned with the behavior of such Boolean functions on input pairs $x, y \in \cube{d}$ that are randomly generated.
\begin{definition}\label{bool:def:correlation}
	For $-1 \leq \alpha \leq 1$ and $x \in \cube{d}$ we let $N_{\alpha}(x)$ denote the distribution over $\cube{d}$ where each component of $y \sim N_{\alpha}(x)$ is independently distributed according to 
	\begin{equation*}
		y_i =
		\begin{cases}
			x_i & \text{with probability } \frac{1 + \alpha}{2}, \\
			-x_i & \text{with probability } \frac{1 - \alpha}{2}. 
		\end{cases}
	\end{equation*}
	We say that $(x, y)$ is randomly $\alpha$-correlated if $x$ is uniformly distributed over $\cube{d}$ and $y \sim N_{\alpha}(x)$.
\end{definition}
Let $\LSH$ denote a distribution over functions $h \colon \{-1, 1\}^d \to R$ where $R$ is a finite set and define 
\begin{equation*}
	p_\alpha = \Pr_{\substack{h \sim \LSH \\ (x, y) \text{ $\alpha$-corr.}}}[h(x) = h(y)].
\end{equation*}
For $0 \leq \beta < \alpha \leq 1$ we wish to characerize the distributions that minimize the expression
\begin{equation} \label{eq:rho}
	\rho_{\alpha, \beta} = \frac{\log(1/p_\alpha)}{\log(1/p_\beta)} 
\end{equation}
when we restrict $\LSH$ to be a distribution over Boolean functions $h \colon \cube{d} \to \{-1, 1\}$.
The expression for $\rho_{\alpha, \beta}$ in equation \eqref{eq:rho} is a well-known quantity in the study of approximate near neighbor search governing the query time and space usage of solutions based on locality-sensitive hashing~\cite{indyk1998}.
\section{Related work}
Indyk and Motwani~\cite{indyk1998} introduced the uniform distribution over the set of dictator functions as a family of locality-sensitive hash functions for the Boolean hypercube.
O'Donnell et al.~\cite{odonnell2014optimal} showed that for general families $\LSH$ it must hold that $\rho_{\alpha, \beta} \geq \log(1/\alpha) / \log(1/\beta)$. 
This matches the upper bound of Indyk and Motwani~\cite{indyk1998} when $\alpha, \beta$ approach $1$. 
Another line of work\cite{panigrahy2008geometric, andoni2016tight} using hypercontractive inequalities showed that $\rho_{\alpha, 0} \geq (1-\alpha)/(1+\alpha)$, 
matching the upper bound of Andoni et al.~\cite{andoni2015optimal}.

The question of finding lower bounds for $\rho_{\alpha, \beta}$ for every choice of $0 \leq \beta < \alpha \leq 1$ is still open. 
In this note we answer the question for distributions over \emph{Boolean} functions, showing that the upper bound of Indyk and Motwani is optimal.
The resulting $\rho$-value is given by
\begin{equation*} 
	\rho_{\alpha, \beta} = \frac{\log((1 + \alpha)/2)}{\log((1 + \beta)/2)}. 
\end{equation*}
\section{Preliminaries}
We will be using tools from the Fourier analysis of Boolean functions to find the minimum of $\rho_{\alpha, \beta}$.
For a more detailed overview we refer to the book by O'Donnell~\cite{odonnell2014analysis}.
We will be using the fact that Boolean functions can be uniquely expressed as multilinear polynomials:
\begin{theorem}
	Every function $f \colon \cube{d} \to \mathbb{R}$ can be uniquely expressed as a multilinear polynomial
	\begin{equation*}
		f(x) = \sum_{S \subseteq [d]}\hat{f}(S)x^S
	\end{equation*}
	where $\hat{f}(S) \in \mathbb{R}$ and $x^S = \prod_{i \in S}x_i$. 
\end{theorem}
For $S \subseteq [d]$ we refer to $\hat{f}(S)$ as the Fourier coefficient of $f$ on $S$. 
The two following Theorems define an inner product between Boolean function and shows how it relates to their Fourier coefficents.  
\begin{theorem}[Plancherel's Theorem]
	For any $f, g \colon \cube{d} \to \mathbb{R}$
	\begin{equation*}
		\ip{f}{g} = \E_{x \sim \cube{d}}[f(x)g(x)] = \sum_{S \subseteq [d]}\hat{f}(S)\hat{g}(S).
	\end{equation*}
\end{theorem}
The concept of Fourier weight will be useful when characterizing the how Boolean functions behave on noisy inputs:
\begin{definition} \label{def:fourier_weight}
	For $f \colon \cube{d} \to R$ define the Fourier weight of $f$ at degree $k \in [d]$ by
	\begin{equation*}
		W^{k}[f] = \sum_{\substack{S \subseteq [d] \\ |S| = k}}\hat{f}(S)^2.
	\end{equation*}
\end{definition}
Consider Plancherel's Theorem with $f = g$ and where $f$ is Boolean-valued.
In this case we get that the sum of the squared Fourier coefficients of $f$ equals 1. 
This result is known as Parseval's Theorem and we will make use of it to determine where to place to Fourier weight of $f$ in order to minimize $\rho$.  
\begin{theorem}[Parseval's Theorem]
	For any $f \colon \cube{d} \to \{-1, 1\}$ 
	\begin{equation*}
		\ip{f}{f} = \E_{x \sim \cube{d}}[f(x)^2] = \sum_{S \subseteq [d]}\hat{f}(S)^2 = \sum_{i = 0}^{d}W^{i}[f] = 1.
	\end{equation*}
\end{theorem}
In order to study the behavior of Boolean functions under noise we introduce the noise operator $T_\alpha$.
\begin{definition}
	For $\alpha \in [-1,1]$ the noise operator with parameter $\alpha$ is the linear operator $T_{\alpha}$ on functions $f \colon \cube{d} \to \mathbb{R}$ defined by 
	\begin{equation*}
		T_{\alpha}f(x) = \E_{y \sim N_{\alpha}(x)}[f(y)].
	\end{equation*}
\end{definition}
The Fourier expansion of $T_{\alpha}f(x)$ is given by $\sum_{S \subseteq [d]}\alpha^{|S|}\hat{f}(S)x^S$. From Plancherel's Theorem it follows that
\begin{align} 	
	\ip{f}{T_\alpha g} &= \E_{x \sim \cube{d}}[f(x) \E_{y \sim N_{\alpha}(x)}[g(y)]] \nonumber \\
	&= \E_{\substack{(x, y) \text{ $\alpha$-corr.}}}[f(x)g(y)] = \sum_{S \subseteq [d]}\alpha^{|S|}\hat{f}(S)\hat{g}(S). \label{eq:noise_plancherel}
\end{align}
In the analysis of our problem the following inequality will be used several times.
For the remainder of this Chapter we will use $\log x$ to denote the natural logarithm of $x$.
\begin{lemma}\label{lem:log_inequality}
	For $x > 0$ we have $\log x \leq x - 1$ with equality if and only if $x = 1$.
\end{lemma}
\section{Bit-sampling is optimal}
Our approach will be to minimize $\rho_{\alpha, \beta}$ subject to the constraint that members of $\LSH$ are Boolean functions $h \colon \cube{d} \to \{-1, 1\}$.
We begin by making some observations to simplify the problem.
For $h \sim \LSH$ we can directly relate the noise-sensitivity under random $\alpha$-correlated inputs to the collision probability.
\begin{align*}
	\E_{\substack{h \sim \LSH \\ (x, y) \text{ $\alpha$-corr.}}}[h(x)h(y)] 
	&= \Pr_{\substack{h \sim \LSH \\ (x, y) \text{ $\alpha$-corr.}}}[h(x) = h(y)] - \Pr_{\substack{h \sim \LSH \\ (x, y) \text{ $\alpha$-corr.}}}[h(x) \neq h(y)] \\
	&= p_\alpha - (1 - p_\alpha) \\
	&= 2p_\alpha - 1.
\end{align*}
Using Equation \eqref{eq:noise_plancherel} we can write $p_\alpha$ as follows:
\begin{equation*}
p_\alpha = (1 + \E_{\substack{h \sim \LSH \\ (x, y) \text{ $\alpha$-corr.}}}[h(x)h(y)])/2 = (1 + \sum_{i = 0}^{d}\alpha^{i} w_i)/2 
\end{equation*}
where we use $w_i$ to denote the expected Fourier weight of $h \sim \LSH$ at degree $i$ defined by $w_i = \E_{h \sim \LSH} \sum_{i = 0}^{d}W^{i}[h]$.
From Plancherel's Theorem we have that $\sum_{i = 0}^{d}w_i = 1$. 
We will now consider how to set $w_{0}, w_{1}, \dots, w_d$ to minimize the expression
\begin{equation*}
	\rho_{\alpha, \beta} = \frac{\log((1 + \sum_{i = 0}^{d}\alpha^{i} w_i)/2)}{\log((1 + \sum_{i = 0}^{d}\beta^{i} w_i)/2)}. 
\end{equation*}
An optimal solution $w^{*}_{0}, \dots, w^{*}_d$ for this problem will yield an optimal solution to the original problem,
provided there actually exists a Boolean-valued function satisfying the weight assignment.
We will show that the assignment $w^{*}_{1} = 1$ and $w^{*}_i = 0$ for $i \neq 1$ minimizes $\rho_{\alpha, \beta}$.
The distribution $\LSH$ therefore only assigns positive probability to functions $h$ that have all their Fourier weight concentrated at degree $1$. 
It turns out that a Boolean function satisfies this weight assignment if and only if it is a dictator function. 
Lemma \ref{lem:dictator} is well-known and is the answer to exercise 1.19 in~\cite{odonnell2014analysis}. 
We include the proof for completeness.
\begin{lemma} \label{lem:dictator}
	Let $f \colon \{-1,1\}^d \to \{-1, 1\}$ and suppose that $W^{1}[f] = 1$, then $f(x) = \pm x_i$.
\end{lemma}
\begin{proof}
	From Parseval's Theorem we know that $\sum_i W^{i}[f] = 1$ and it follows that $\hat{f}(S) = 0$ for $|S| \neq 1$. 
	The function $f$ can therefore be written in the form $f(x) = \sum_{i = 1}^{d}\hat{f}_i x_i$ where $\hat{f}_i = \hat{f}(S)$ for $S = \{i\}$.
	By the condition $W^{1}[f] = 1$ there exists $j \in [d]$ such that $\hat{f}_j \neq 0$. 
	Fix the $d-1$ components $x_{i \neq j}$ of $x$ and note that since $f$ maps to $\{-1, 1\}$ the sum $f(x) = \hat{f}_j x_j + \sum_{i \neq j}\hat{f}_i x_i$ must satisfy $f(x) = \pm1$ when $x_j = \pm1$.
	For $\hat{f}_j \neq 0$ this is only possible when $\hat{f}_j = \pm 1$ which implies that $\hat{f}_{i} = 0$ for $i \neq j$.
	It follows that $f$ must be one of the $2d$ functions of the form $f(x) = \pm x_i$.
\end{proof}
\subsection{Optimal Fourier weight at degree zero}
We begin by arguing that we can restrict our attention to showing that dictator functions are optimal in the case where $0 < \beta < \alpha < 1$.
If $\alpha = 1$ then for $w_{1} = 1$ we have that $\rho = 0$ which is the best we can hope for (but this could also be achieved by other weight assignments, hence the statement of the main theorem is for $\alpha < 1$.).
For $\beta = 0$ the following Lemma showing that $w_{0}^{*} = 0$ combined with the fact that for this setting we maximize $p_\alpha$ by setting $w_{1} = 1$ shows that the dictator functions are optimal.
We will now show that an optimal solution has no Fourier weight at degree zero.
\begin{lemma} \label{lem:w_0}
	$w_{0}^{*} = 0$.
\end{lemma}
\begin{proof}
	If $w_{0} = 1$ we have $\rho = 1$ and it is clear that $\rho < 1$ if we set $w_{1} = 1$. 
	Suppose that $0 < w_{0}^{*} < 1$. 
	We will show that in this case we can move some weight from $w_{0}$ to $w_{1}$ and decrease the value of $\rho$.
	For a given weight assignment define $s(\alpha) = \sum_{i} \alpha^i w_i$ and write $w_{1}$ as $w_{1} = 1 - \sum_{j \neq i}w_j$. 
	The partial derivative of $\rho = \log((1 + s(\alpha))/2)/\log((1 + s(\beta))/2)$ with respect to $w_{0}$ is given by
	\begin{equation*}
		\frac{\partial \rho}{\partial w_{0}} = \frac{\frac{\partial s(\alpha) / \partial w_{0}}{1 + s(\alpha)}\log \frac{1 + s(\beta)}{2} 
		- \frac{\partial s(\beta) / \partial w_{0}}{1 + s(\beta)}\log \frac{1 + s(\alpha)}{2}}
		{\log^2 \frac{1 + s(\beta)}{2}}.
	\end{equation*}
	By rearranging and using that $\partial s(\alpha) / \partial w_{0} = 1 - \alpha$ we find that $\frac{\partial \rho}{\partial w_{0}} > 0$ is equivalent to  
	\begin{equation*} 
		\frac{1 + s(\beta)}{1 - \beta} \log \frac{1 + s(\beta)}{2} > \frac{1 + s(\alpha)}{1 - \alpha} \log \frac{1 + s(\alpha)}{2}.
	\end{equation*}
	It suffices to show that the function $g(x) = \frac{1 + s(x)}{1 - x} \log \frac{1 + s(x)}{2}$ is decreasing for $0 < x < 1$. 
	\begin{equation} \label{eq:condition}
		\frac{\partial g}{\partial x} = \frac{s'(x)(1 - x) + (1 + s(x))}{(1-x)^2} \log \frac{1 + s(x)}{2} + \frac{s'(x)}{1 - x}.
	\end{equation}
	Rewriting, this is equivalent to showing that 
	\begin{equation*}
		(s'(x)(1-x) + 1 + s(x)) \log \frac{1 + s(x)}{2} + (1-x)s'(x) < 0.
	\end{equation*}
	By the assumption that $0 < w_0 < 1$ we have that $0 < s(x) < 1$ and using Lemma \ref{lem:log_inequality} we get that $\log((1 + s(x))/2) < (s(x) - 1)/2$.
	The condition in equation \eqref{eq:condition} then simplifies to showing that $s'(x)(1-x) + s(x) \leq 1$.
	The function $s(x) = \sum_i w_i x^i$ is a weighted sum of simple monomials where the weights sum to one.
	It therefore suffices to show that the inequality holds for every monomial $s_k(x) = x^k$ where $k = \{0, 1, \dots, d \}$.
	For $k = 0$ and $k = 1$ we have $s_{k}'(x)(1-x) + s_{k}(x) = 1$ satisfying the desired inequality.
	For $k \geq 2$ we have $s_{k}'(x)(1-x) + s_{k}(x) = kx^{k-1} + (k-1)x^k$.
	We see that $s_{k}(0) = 0$ and $s_{k}(1) = 1$ and by inspecting the derivative of $s_{k}(x)$ we see that it is increasing for $x \in (0,1)$. 
	It follows that the inequality is satisfied, completing the proof.
\end{proof}
\subsection{A continuous optimization problem}
In order to simplify the problem of minimizing $\rho$ we will optimize over a larger space.
In particular we will let $W$ denote a collection of pairs $(w, \kappa)$ such that $\sum_{w \in W} w = 1$ where we restrict $\kappa \in \mathbb{R}$ to satisfy $\kappa \geq 1$.
We define $s(x) = \sum_{(w, \kappa) \in W} w x^\kappa$ and we will now attempt to specify the function $s$ that minimizes  
\begin{equation*}
	\rho_{\alpha, \beta} = \frac{\log \frac{1 + s(\alpha)}{2}}{\log \frac{1 + s(\beta)}{2}}
\end{equation*}
subject to the constraint that $s(\beta) = b \leq \beta$ is fixed. 
The constraint that $s(\beta) \leq \beta$ follows from the restrictions on $s$.
We can therefore write $b = \beta^\gamma$ for some $\gamma \geq 1$.
For fixed $s(\beta)$ it is clear that we minimize $\rho$ by maximizing $s(\alpha)$. 
\begin{lemma} \label{lem:gamma}
	For fixed $s(\beta) = \beta^\gamma$ we maximize $s(\alpha)$ by setting $s(x) = x^\gamma$.
\end{lemma}
\begin{proof}
Let $w$ denote the weight on the exponent $\gamma$ in the specification $W$ of $s$. 
We will prove that if $w < 1$ then we can increase $s(\alpha)$ by rearranging the weights of $s$ to put more weight onto $(w, \gamma)$.
Note that if $w < 1$ and we have a valid configuration of weights (in the sense that $s(\beta) = \beta^\gamma$) there must exist exponents $\gamma_{0} < \gamma < \gamma_{1}$ such that there is positive weight on $\gamma_{0}$ and $\gamma_{1}$. 
If all the remaining weight was concentrated to either side of $\gamma$ the condition $s(\beta) = \beta^\gamma$ would be violated.
We will now move $\varepsilon_{0}$ weight from $w_{0}$ to $w$ and $\varepsilon_{1}$ weight from $w_{1}$ to $w$ where we set $\varepsilon_{0}, \varepsilon_{1}$ to ensure that $s(\beta) = \beta^\gamma$ after the move.
It turns out that this condition is satisfied for the following ratio
\begin{equation*}
	\varphi(\beta) = \varepsilon_{1} / \varepsilon_{0} = \frac{\beta^{\gamma_{0}} - \beta^\gamma}{\beta^\gamma - \beta^{\gamma_{1}}} > 0.
\end{equation*}
The change in $s(\alpha)$ due to the rearrangement of weights can be shown to be positive if $\varphi(\alpha) < \varphi(\beta)$.
Therefore, it suffices to show that $\varphi(x)$ is decreasing for $0 < x < 1$ when $\gamma_{0} < \gamma < \gamma_{1}$.
To simplify further, we define $\lambda_{0} = \gamma_{0} - \gamma$ and $\lambda_{1} = \gamma_{1} - \gamma$ which satisfy $\lambda_{0} < 0 < \lambda_{1}$.
Rewriting $\varphi(x) = -(1 - x^\lambda_{0})/(1 - x^\lambda_{1})$ and differentiating we get
\begin{align*}
	\frac{\partial \varphi}{\partial x} &= \lambda_{0} x^{\lambda_{0}}(1 - x^{\lambda_{1}} - \lambda_{1} x^{\lambda_{1}}(1-x^{\lambda_{0}}) < 0 \\
										&\iff \frac{\lambda_{0} x^{\lambda_{0}}}{1 - x^{\lambda_{0}}} > \frac{\lambda_{1} x^{\lambda_{1}}}{1 - x^{\lambda_{1}}}. 
\end{align*}
It suffices to show that $\psi(x) = \frac{xa^x}{1- a^x}$ is decreasing in $x$ for $a \in (0,1)$.
We have that $\psi'(x) = a^{x}(1 - a^{x}) + a^x \log a^x$. 
Define $z = a^x$ and note that $z > 0$ and $z \neq 1$. 
We have that $z(1-z) + z \log z < 0 \iff (1 - z) + \log z < 0$ and by Lemma \ref{lem:log_inequality} we see that $\log z < z - 1$, completing the proof. 
\end{proof}
\subsection{Univariate analysis}
According to Lemma \ref{lem:gamma} we can now restrict our attention to the problem of finding $\gamma \geq 1$ that minimizes the function
\begin{equation*}
	\rho(\gamma) = \frac{\log \frac{1 + \alpha^\gamma}{2}}{\log \frac{1 + \beta^\gamma}{2}} \,.
\end{equation*}
We will show the derivative of $\rho$ is positive, implying that it is minimized when $\gamma = 1$.
\begin{lemma} \label{lem:derivative_gamma}
	$\rho'(\gamma) > 0$.
\end{lemma}
\begin{proof}
From inspecting the derivative of $\rho$ with respect to $\gamma$ we see that
\begin{align*}
&\frac{\partial \rho}{\partial \gamma} > 0 \\
		&\iff \frac{\alpha^{\gamma} \log \alpha}{1 + \alpha^{\gamma}} \log \frac{1 + \beta^\gamma}{2} - \frac{\beta^{\gamma} \log \beta}{1 + \beta^{\gamma}} \log \frac{1 + \alpha^\gamma}{2} > 0 \\ 
		&\iff \frac{1 + \beta^{\gamma}}{\beta^{\gamma} \log \beta} \log \frac{1 + \beta^\gamma}{2} >  \frac{1 + \alpha^{\gamma}}{\alpha^{\gamma} \log \alpha} \log \frac{1 + \alpha^\gamma}{2}.  
\end{align*}
Therefore it suffices to show that the function $g(x) = \frac{1 + x^{\gamma}}{x^{\gamma} \log x} \log \frac{1 + x^\gamma}{2}$ is decreasing for $0 < x < 1$ and $\gamma \geq 1$.
From inspecting $g'(x)$ we see that the condition that $g'(x) < 0$ is equivalent to
\begin{equation*}
-(1 + x^\gamma + \gamma \log x) \log \frac{1 + x^\gamma}{2} + \gamma x^{\gamma} \log x < 0 
\end{equation*}
If $1 + x^\gamma + \gamma \log x \geq 0$ then the condition is satisfied and we are done.
Otherwise we can use the fact that $-(1 + x^\gamma + \gamma \log x) > 0$ together with Lemma \ref{lem:log_inequality} to produce following derivation: 
\begin{align*}
&-(1 + x^\gamma + \gamma \log x) \log \frac{1 + x^\gamma}{2} + \gamma x^{\gamma} \log x \\
&\qquad < -(1 + x^\gamma + \log x^\gamma)\frac{x^\gamma - 1}{2} + x^\gamma \log x^\gamma \\
&\qquad = 1 - x^{2\gamma} + (1 + x^\gamma) \log x^\gamma 
\end{align*}
Reapplying Lemma \ref{lem:log_inequality} we see that $(1 + x^\gamma) \log x^\gamma < (1 + x^\gamma)(x^\gamma - 1) = -(1 - x^{2\gamma})$ completing the proof. 
\end{proof}
\subsection{Stating the result}
We will now summarize how the results from the previous subsections yield the main result of this paper as stated in the abstract. 
To find the the distribution over Boolean functions minimizing $\rho$ we first considered the optimal weight assignment in the expression $s(x) = \sum_i w_i \alpha^i$ subject to the constraint that $\sum_i w_i = 1$. 
Finding an optimal assignment does not guarantee that we have solved the problem, because there may not exist a Boolean function with a given weight assignment, 
but if one or more Boolean functions that satisfy the optimal assignment exists we will have solved the problem. 
In Lemma \ref{lem:w_0} we showed that an optimal solution $w_0^*, w_1^*, \dots w_d^*$ must have $w_0^* = 0$.
Therefore the optimal solution can only have non-zero weight on exponents $k \geq 1$.
Next, in Lemma \ref{lem:gamma}, we argued that if we allow continuous exponents $k \in \mathbb{R}$ with $k \geq 1$ in $s(x)$ then the problem of minimizing $\rho$ becomes the problem of selecting $\gamma \geq 1$ where $s(x) = x^\gamma$.
Lemma \ref{lem:derivative_gamma} showed that $\rho(\gamma)$ is increasing, so to minimize $\rho$ we want to set $\gamma = 1$.
The conclusion from these optimization problems is that we minimize $\rho$ by setting $w_1^* = 1$.
Finally Lemma \ref{lem:dictator} shows that the subset of the Boolean functions with $w_1 = 1$ is exactly the set of dictator functions $f(x) = \pm x_i$. 
Together with the fact that $w_1^* = 1$ is a unique minimum of $\rho$ in the weight assignment problem we get our main result.
\section{Open problems}
\paragraph{Orthogonal search.}
It appears that the same techniques can be used to show that pairs of functions of the form $f(x) = x_i x_j$, $g(y) = -x_i x_j$ minimize the function
\begin{equation*}
	\frac{\log(1/\min(p_{\alpha},p_{-\alpha}))}{\log(1/\max(p_{\beta}, p_{-\beta}))}.
\end{equation*}

\paragraph{Extension to negative correlation.}
It seems likely that the dictator functions or bit-sampling minimizes $\rho$ for the entire interval $-1 \leq \beta < \alpha \leq 1$. 
Unfortunately the current proof breaks down in places.

\paragraph{General hash functions.}
Showing tight bounds for hash function with an arbitrary range is an interesting open problem.
For orthogonal search this is an open problem even in the case of $\rho_{\alpha, 0}$.
For more information see the symmetric Gaussian problem in \cite{odonnell2012open}. 

Investigating what the implications of the results in this paper for functions with an arbitrary range through the use of $1$-bit hashing is an interesting problem.

\bibliography{bool}
\bibliographystyle{abbrv}
\end{document}